\documentclass[11pt]{article}
\usepackage{CSTheoryToolkitCMUStyle}
\usepackage{cleveref}
\usepackage{mdframed}

\usepackage{float} 
\floatstyle{plain} 
\newfloat{protocol}{htbp}{lop} 
\floatname{protocol}{Protocol} 
\newfloat{program}{htbp}{lop} 
\floatname{program}{Program} 
\usepackage{authblk}

\begin{document}

\title{Quantum One-Time Memories from Stateless Hardware, Random Access Codes, and Simple Nonconvex Optimization}

\author{Lev Stambler\thanks{\href{mailto:levstamb@umd.edu}{levstamb@umd.edu}}}
\affil{Joint Center for Quantum Information and Computer Science,\\
University of Maryland, College Park, MD 20742, USA}
\affil{Department of Computer Science, University of Maryland, MD 20742, USA}
\affil{NeverLocal Ltd. London, WC2H 9JQ, UK}

\date{\today}
\maketitle
\newcommand{\advers}{\mathcal{A}}
\newcommand{\Sim}{\mathcal{S}}
\newcommand{\SimOTM}{\mathcal{S}_{OTM}}
\newcommand{\cryptd}[1]{\widetilde{#1}}
\newcommand{\msg}{\text{msg}}
\newcommand{\msgSize}{m}
\newcommand{\msgSpace}{\set{0, 1}^\msgSize}
\newcommand{\aux}{\texttt{aux}}
\newcommand{\PRG}{\text{PRG}}
\newcommand{\randomGets}{\xleftarrow{\$}}
\newcommand{\negl}{\text{negl}}
\newcommand{\orac}{\calO}
\newcommand{\oracP}{\widetilde{\calO}}
\newcommand{\compInd}{\approx_c}
\newcommand{\Hyb}{\textbf{Hyb}}
\newcommand{\out}{\texttt{out}}
\newcommand{\Msgs}{\texttt{Msgs}}

\newcommand{\Hmin}{H_{\text{min}}}
\newcommand{\HHill}{H_{\text{HILL}}}
\newcommand{\HMax}{H_{\text{max}}}
\newcommand{\HMaxc}{H_{\text{max}}^c}

\newcommand{\Obf}{\text{Obf}}
\newcommand{\Obfd}[1]{\cryptd{#1}}

\newcommand{\OTMPrep}{\texttt{OTM.Prep}}
\newcommand{\OTMRead}{\texttt{OTM.Read}}

\newcommand{\Iacc}{\text{I}_{\text{acc}}}

\newcommand{\racEnc}[1]{\calE\left(#1\right)}

\newcommand{\OObfd}[1]{\overline{#1}}
\newcommand{\measChan}{\mathcal{M}}
\newcommand{\Iden}{\mathbb{I}}

\newcommand{\LEV}[1]{\textcolor{red}{LEV: #1}}
\newcommand{\ChatGPT}[1]{\textcolor{blue}{TODO ChatGPT: #1}}

\newcommand{\pBad}{0.83}

\begin{abstract}
	We present a construction of one-time memories (OTMs) using classical-accessible stateless hardware, building upon the work of Broadbent et al.\ and Behera et al.\ \cite{Broadbent_2021, behera2021noise}.
	Unlike the aforementioned work, our approach leverages quantum random access codes (QRACs) to encode two classical bits, \( b_0 \) and \( b_1 \), into a single qubit state \( \racEnc{b_0 b_1} \) where the receiver can retrieve one of the bits with a certain probability of error.  

	To prove soundness, we define a nonconvex optimization problem over POVMs on \( \mathbb{C}^2 \). This optimization gives an upper bound on the probability of distinguishing bit \( b_{1-\alpha} \) given that the probability that the receiver recovers bit $b_\alpha$ is high.
	Assuming the optimization is sufficiently accurate, we then prove soundness against a polynomial number of classical queries to the hardware.
\end{abstract}

\section{Introduction}
One-time programs (and variants) were first formally introduced in Ref.~\cite{goldwasser2008one} and sit at the top of the cryptographic wish list.
Unfortunately, one-time programs are impossible to achieve in the standard model and idealized quantum model \cite{goldwasser2008one, Broadbent2012QuantumOP}.
Still, if hardware assumptions enable a primitive known as one-time memories, then one-time programs can be built \cite{goldwasser2008one}.
Ref.~\cite{Broadbent2012QuantumOP} also showed how to build (classical and quantum) one-time programs from OTMs in the Universal Composability (UC) model
A one-time memory (OTM) is a cryptographic primitive which can be used to build one-time programs and non-interactive secure two party computation \cite{Yao1982ProtocolsFS, goyal2010founding}.
OTMs can be thought of as a non-interactive version of oblivious transfer (OT) where a sending party, Alice, sends a one-time memory to a receiving party, Bob.
The memory encodes two classical strings, $m_0$ and $m_1$, and Bob can only learn one of the strings.
After Bob reads $m_b$, any encoding of $m_{1 - b}$ is ``erased'' and Bob cannot recover $m_{1 - b}$.
Currently, OTMs rely on hardware-specific assumptions such as tamper-proof hardware \cite{smid1981integrating, Katz2007UniversallyCM} and trusted execution environments (TEEs) \cite{Perrig2009ReducingTT, Zhao2019OneTimePM}.
Unfortunately, so-called ``replay'' attacks have been carried out on TEE based OTMs \cite{skarlatos2019microscope}.
A series of works \cite{liu2014single, liu2015privacy, liu2023depth, barhoush2023powerful} make a separate set of assumptions on quantum hardware to build one-time memories. 

Ref.~\cite{Broadbent_2021} began a line of work on constructing one-time memories assuming \emph{stateless} trusted hardware. 
Indeed, if we cannot immunize trusted hardware against replay attacks, then we can use trusted hardware in conjunction with quantum states to build one-time memories.
Importantly, to build one-time memories from (stateless) trusted hardware and quantum states, we must assume that the adversary cannot make \emph{quantum} queries to the hardware \cite{Broadbent_2021}.
Ref.~\cite{Broadbent_2021} proposes a protocol using Wiesner coding to build one-time memories from stateless hardware and quantum states but only prove security for a \emph{linear} number of queries to the hardware.
Behera et al.\ \cite{behera2021noise} provide a construction for noise-tolerant quantum tokens using Wiesner coding and stateless hardware.
Ref.~\cite{behera2021noise} then prove security for a \emph{polynomial} number of queries to the stateless hardware as well as provide a secure and noise-tolerant construction for one-time memories.	

Moreover, Chung et al.\ \cite{chung2019cryptography} construct disposable backdoors using random subspace states and stateless hardware.
Ref.~\cite{chung2019cryptography} also provides a construction for one-time memories using random subspace states and stateless hardware.

In a more complexity-theoretic focused work, Li et al.\ \cite{li2023classical} show that $\mathsf{BQP/qploy} \neq \mathsf{BQP/poly}$ and $\mathsf{QMA} \neq \mathsf{QCMA}$ relative to an \emph{classical}-accessible oracle (where the oracle can only be queried classically).
We note that this is a similar setting in our work where we model the stateless hardware as a classical-accessible oracle.
Interestingly, the existance of one-time programs in the presence of stateless hardware seems to provide a plausible seperation between providing a classical-accessible oracle and a quantum-accessible oracle: as noted by Ref.~\cite{Broadbent2012QuantumOP}, one-time programs are impossible with quantum-accessible oracles, but in this work we construct one-time memories with classical-accessible oracles.
We leave further exploration of this connection to future work.

\subsection{Main Result}
Our main result is the construction of one-time memories from stateless hardware and quantum random access codes (QRACs) \cite{ambainis1999dense} in lieu of Wiesner states.
Moreover, one of our key contributions is a simple proof of soundness which leverages a nonconvex optimization over POVMs for states in $\C^2$ (single qubit states)\footnote{
	We note that an $\epsilon$-net technique in combination with brute force search could potentially be used to solve our optimization problem with provable guarantees.
	We leave this as an open question for future work.
}.
Specifically, we use QRACs to encode two classical bits, $b_0, b_1$, into a single qubit state, $\racEnc{b_0 b_1}$, where the receiver can recover one of the bits with some probability of error.

We first provide a simple nonconvex optimization which upper-bounds the trace distance between $\racEnc{b_0 0}$ and $\racEnc{b_0 1}$ given that the receiver performed some POVMs which can recover $b_0$ with high probability.
Then, to construct our OTMs for messages $m_0, m_1$, we sample 2 $n$-bit strings, $r_0, r_1$.
We encode $r_0, r_1$ into $n$ QRAC states, $\racEnc{r_0^1 r_1^1} \otimes \ldots \otimes \racEnc{r_0^n r_1^n}$.
To ``unlock'' message $m_b$, the receiver must recover a certain fraction of $r_b$.
We then show that in order to recover $r_b$, the receiver must ``collapse'' the ability to distinguish between $\racEnc{r_b^i 0}$ and $\racEnc{r_b^i 1}$ for most $i \in [n]$.

For simplicty, we will model the trusted hardware as a black box which only accepts classical queries.

\section{Preliminaries}
\label{sec:prelims}
In this section, we will introduce the necessary background for our construction of one-time memories.

\subsection{Quantum Random Access Code (QRAC)}
\label{subsec:qrac}
A quantum random access code is a uniquely quantum primitive that allows a sender to encode multiple messages into a single quantum state.
A receiver can then perform a measurement on the state to recover one of the messages with some probability of error.
\begin{definition}[$2 \mapsto 1$ Quantum Random Access Code (QRAC) \cite{ambainis1999dense}]
	\label{def:qrac}
	A QRAC is an ordered tuple $(\calE, \mu_0, \mu_1)$ where $\calE : \F_2^2 \rightarrow \C^2$ and $2$ orthonormal measurements $\mu_\alpha = \set{\ket{\phi_\alpha^0}, \ket{\phi_\alpha^1}}$ for $\alpha \in \set{0, 1}$.
\end{definition}

For our purposes, we will consider a fixed $2 \mapsto 1$ QRAC.
Specifically, we will use the optimal $2 \mapsto 1$ QRAC from Ref.~\cite{ambainis1999dense}.
Define $\ket{\psi_\theta} = \cos(\theta) \ket{0} + \sin(\theta) \ket{1}$.
Then, the optimal $2 \mapsto 1$ QRAC is given by the following mapping:
\begin{align*}
\racEnc{00} \mapsto \ket{\psi_{\pi/8}}\bra{\psi_{\pi/8}} \;\; \; \;
&\racEnc{01} \mapsto \ket{\psi_{-\pi/8}} \bra{\psi_{-\pi/8}} \\
\racEnc{10} \mapsto \ket{\psi_{5\pi/8}} \bra{\psi_{5\pi/8}}\;\;
&\racEnc{11} \mapsto \ket{\psi_{-5\pi/8}} \bra{\psi_{-5\pi/8}}\;\;
\end{align*}
and the measurements are given by projecting onto $\mu_0 = \set{\ket{0}\bra{0}, \ket{1}\bra{1}}$ (the $Z$ basis)
and $\mu_1 = \set{\ket{\psi_{\pi/4}}\bra{\psi_{\pi/4}}, \ket{\psi_{-\pi/4}}\bra{\psi_{-\pi/4}}}$ (the $X$ basis).

We also know that the optimal success probability of this QRAC is $\cos^2(\pi / 8) < 0.854$ \cite{ambainis1999dense}.

\subsection{POVMs, Trace Distance, and Distinguishing Probability}
We will use a few basic definitions and theorems from quantum information theory.
For a more in-depth treatment, see Nielsen and Chuang \cite{nielsen2001quantum}.

\begin{definition}[POVM]
	Let $\calH$ be a Hilbert space. A Positive Operator-Valued Measure (POVM) on $\calH$ is a collection of operators $\set{M_i}_{i \in [m]}$ such that $M_i \geq 0$ for all $i$ and $\sum_{i=1}^m M_i = I$.
\end{definition}

POVMs are useful as they can represent any quantum measurement, including projective measurements.

\begin{proposition}[POVMs Capture All Measurements]
	Every quantum measurement can be represented as a POVM. In particular, for any projective measurement with measurement operators $P_i$, there exists a POVM $\{M_i\}$ such that $P_i^\dagger P_i = M_i$.
\end{proposition}

Importantly, given a POVM, we can find the set of projectors that correspond to the POVM by using Naimark's Dilation Theorem.

\begin{theorem}[Naimark's Dilation Theorem, \cite{naimark1943representation}]
	Every POVM can be realized as a projective measurement in a higher-dimensional Hilbert space. Formally, given a POVM $\{M_i\}_{i \in [m]}$ on $\mathcal{H}$, there exists an auxiliary Hilbert space $\mathcal{H}_a$, an isometry $V : \mathcal{H} \to \mathcal{H} \otimes \mathcal{H}_a$, and projectors $\{P_i\}_{i \in [m]}$ on $\mathcal{H} \otimes \mathcal{H}_a$ such that:
	\[
		M_i = V^\dagger P_i V.
	\]

	Alternatively, the projectors can be constructed as:
	\[
		P_i = U \sqrt{M_i},
	\]
	where $U$ is an isometry acting on $\mathcal{H} \otimes \mathcal{H}_a$.
\end{theorem}

Finally, we will use the trace distance to quantify the distinguishability between quantum states.

\begin{definition}[Trace Distance]
	The trace distance between two quantum states $\rho$ and $\sigma$ is defined as:
	\[
		T(\rho, \sigma) =  \frac{1}{2} \|\rho - \sigma\|_1,
	\]
	where $\|\cdot\|_1$ denotes the trace norm, $\|\rho - \sigma\|_1 = \mathrm{Tr} \sqrt{(\rho - \sigma)^\dagger (\rho - \sigma)}$.
\end{definition}

For our purposes, trace distance is useful as it directly gives an upper bound on the probability of distinguishing two quantum states.

\begin{proposition}[Trace Distance and Distinguishability]
	The trace distance $T(\rho, \sigma)$ quantifies the distinguishability between two quantum states. Specifically, for any POVM $\{M_i\}$, the maximum probability of distinguishing $\rho$ and $\sigma$ is:
	\[
		P_{\text{dist}} = \frac{1}{2} (1 + T(\rho, \sigma)).
	\]
\end{proposition}

%

\subsection{One Time Memories}
Though more complex definitions of one-time memories exist (such as the UC-based definitions in Ref.~\cite{Broadbent2012QuantumOP}), we will use a simpler definition which is very similar to Ref.~\cite{liu2023depth} for the purposes of this paper.

\begin{definition}[One-Time Memory]
	\label{def:otmcorr}
	A one-time memory is a protocol between a sender and receiver which can be represented as a tuple of algorithms $(\OTMPrep, \OTMRead)$ where
	\begin{itemize}
		\item $\OTMPrep$ is a probabilistic algorithm which takes as input $m_0, m_1 \in \Msgs$ and outputs a quantum state $\rho$ as well as classical auxiliary information $\aux$.
		\item $\OTMRead$ is a (potentially probabilistic) algorithm which takes as input $\rho, \aux$ and $\alpha \in \set{0, 1}$ and outputs a classical string $m_\alpha$ with probability $1 - \delta$.
	\end{itemize}
\end{definition}

\begin{definition}[One-Time Memory Soundness]
	\label{def:otmSound}
	A one-time memory $(\OTMPrep, \OTMRead)$ is sound relative to an adversary $\calA$ if there is a simulator $\Sim$ and a negligible function $\gamma(\cdot)$ for every $m_0, m_1 \in \Msgs$. The simulator $\Sim$ makes at most one classical query to $g^{m_0, m_1}: \set{0, 1} \rightarrow \set{m_0, m_1}$ (where $g(\alpha) = m_\alpha$). The requirement is:

	\[
		\calA\left(\OTMPrep(1^{\lambda}, m_0, m_1)\right) \overset{\gamma(\lambda)}{\approx_c} \SimOTM^{g^{m_0, m_1}}\left(1^{\lambda} \right),
	\]

	where $\overset{\gamma(\lambda)}{\approx_c}$ denotes computational indistinguishability for a negligible function $\gamma(\lambda)$.
\end{definition}

\newcommand{\sciPyVal}{0.253}
\newcommand{\sciPyUpper}{0.3}
\newcommand{\guessProb}{0.65}
\newcommand{\guessProbInv}{0.2}
\newcommand{\pAcc}{0.85}
\newcommand{\pMax}{0.854}
\newcommand{\wtb}{c_0}
\newcommand{\wtbOne}{c_1}

\section{Distinguishing After Measurement}
In this section, we outline the non-convex program to get a bound on the trace distance between two QRAC states after some measurement.
Without loss of generality, we assume that the receiver can recover bit $b_0$ with high probability.
Then, we bound the trace distance between the states $\racEnc{b_0 0}$ and $\racEnc{b_0 1}$ after the receiver measures the state $\racEnc{b_0 b_1}$ and recovers $b_0$.
Importantly, this bound on trace distance only holds \emph{when the probability of recovering $b_0$ is high}.

\begin{program}[H]
	\begin{mdframed}
		Optimize POVM $E_0, E_1$ to maximize
		\begin{align*}
			\max_{b_0} \;
			T\bigg(&
				{\sqrt{E_{0}} \racEnc{b_0 0} \sqrt{E_{0}}^\dagger} + {\sqrt{E_{1}} \racEnc{b_0 0} \sqrt{E_{1}}^\dagger},  \\
			       &\; {\sqrt{E_{0}} \racEnc{b_0 1} \sqrt{E_{0}}^\dagger} + {\sqrt{E_{1}} \racEnc{b_0 1} \sqrt{E_{1}}^\dagger}
		       \bigg)
			\end{align*}
			for $b_0 \in \set{0, 1}$ subject to
			\begin{itemize}
				\item $E_0$ and $E_1$ form a POVM
				\item For $\rho_0 = \half \left(\racEnc{0 0} + \racEnc{0 1}\right)$ and $\rho_1 = \half\left(\racEnc{1 0} + \racEnc{1 1}\right)$,
					\[
						\frac{1}{2}\left(\Tr[E_0 \rho_0] + \Tr[E_1 \rho_1]\right) \geq \pBad.
					\]
			\end{itemize}
		\end{mdframed}
		\caption{Optimization to Bound Cryptographic Disturbance}
		\label{fig:disturbance}
	\end{program}
	\begin{claim}
		\label{claim:nonconv}
		Let $\measChan$ be a quantum channel, $\aux$ be a quantum state independent of $b_1$, and $b_0, b_1$ be uniformly random bits.
		If $\Pr_{b_0, b_1, \measChan}[b_0 = \wtb \mid \wtb, \rho \gets \measChan(\racEnc{b_0 b_1})] \geq \pBad$, then
		\[
			T(\measChan(\racEnc{b_0 b_1}) \otimes \aux, \measChan(\racEnc{b_0 (1 - b_1)}) \otimes \aux) \leq \sciPyUpper.
		\]
	\end{claim}
	We can then prove the above claim assuming that the non-convex optimization problem of program~\ref{fig:disturbance} gives an ``accurate enough'' result as captured by the following conjecture.
	\begin{conjecture}
		\label{conj:nonconv}
		The optimal value of the optimization problem in program~\ref{fig:disturbance} is upper-bounded by $\sciPyUpper$.
		\footnote{
			We use two separate optimization techniques from SciPy's optimize library \cite{2020SciPy-NMeth}: \texttt{minimize} with\texttt{ SLSQP} and \texttt{differential evolution}.
			For the \texttt{minimize} method, we run the optimization $1,000$ times using random initializations of $E_0$ ($E_1$ then equals $I - E_0$).
			For differential evolution, we set the maximum number of iterations to $10,000$.
			Our SciPy non-convex program gives a value of $\sciPyVal$ but we round up to $\sciPyUpper$ for simplicity/ to account for potential error.
			The code for the optimization is available at this \href{https://gist.github.com/Lev-Stambler/9ca4a571d85d714f7c7589ebb6ea1b09}{GitHub Gist}.
		}
	\end{conjecture}

	\begin{proof}[Proof of \cref{claim:nonconv}]
		First, note that if the adversary has $3$ or more POVMs, $E_0, E_1, \dots, E_\ell$, then the adversary must run some post-processing on the output of the channel to get a single bit.
		Thus, we can absorb the post-processing into a computation described by two POVMs.
		So then, model the output $\wtb$ as the output of POVMs $E_0, E_1$ and the channel $\measChan$ as returning the POVM result, $c_0$, and the post-measurement state, $\rho$.

		Then if $\rho, \wtb \gets \measChan(\racEnc{b_0 b_1})$ and $\Pr_{b_0, b_1, \measChan}[\wtb = b_0 \mid \rho, \wtb \gets \measChan(\racEnc{b_0 b_1})] \geq \pBad$, we have that
		\begin{align*}
			\Pr_{b_0, b_1, \measChan}[\wtb = b_0 \mid \wtb \gets \measChan(\racEnc{b_0 b_1})] &= \E_{b_0}\left[\Tr[E_{b_0} \rho_{b_0}]\right] \\
													  &= \Pr[b_0 = 0] \cdot \Tr[E_0 \rho_0] + \Pr[b_0 = 1] \Tr[E_1 \rho_1] \\
													  &= \frac{1}{2} \left(\Tr[E_0 \rho_0] + \Tr[E_1 \rho_1]\right)
		\end{align*}
		where $\rho_0$ and $\rho_1$ are defined in program~\ref{fig:disturbance}.
		Then, let $W$ be some isometry and
		\[
			M_{\wtb} = W \sqrt{E_{\wtb}}.
		\]
		By Naimark's Dialation Theorem \cite{naimark1943representation}, we have that our post-measurement state after $\measChan$ returns $\wtb$ is
		\[
			\frac{
				M_{\wtb} \; \racEnc{b_0 b_1} \; M_{\wtb}^\dagger.
			}{
				\Tr[E_{\wtb} \racEnc{b_0 b_1}]
			}.
		\]
		Also, we can note that the probability of the output of $\measChan$ returning $\wtb$ equals $\Tr[E_{\wtb} \racEnc{b_0 b_1}]$.
		So, the mixed state after the channel $\measChan$ is
		\begin{align*}
			\sigma_{b_0 b_1} = &\Pr[\wtb = 0 \mid \measChan(b_0 b_1)] \cdot \frac{
				M_{0} \; \racEnc{b_0 b_1} \; M_{0}^\dagger.
			}{
				\Tr[E_{0} \racEnc{b_0 b_1}]
			} + 
			\Pr[\wtb = 1 \mid \measChan(b_0 b_1)] \cdot \frac{
				M_{1} \; \racEnc{b_0 b_1} \; M_{1}^\dagger.
			}{
				\Tr[E_{1} \racEnc{b_0 b_1}]
			} \\
			&= 
			M_0 \; \racEnc{b_0 b_1} \; M_0^\dagger + M_1 \; \racEnc{b_0 b_1} \; M_1^\dagger.
\end{align*}

Then, by the upper bound on the optimization problem in program~\ref{fig:disturbance} and \cref{conj:nonconv}, we have that
\[
	T(\sigma_{b_0 0}, \sigma_{b_0 1}) = T\left(
		M_0 \; \racEnc{b_0 0} \; M_0^\dagger + M_1 \; \racEnc{b_0 0} \; M_1^\dagger,
		M_0 \; \racEnc{b_0 1} \; M_0^\dagger + M_1 \; \racEnc{b_0 1} \; M_1^\dagger,
	\right) \leq \sciPyUpper.
\]
as the trace distance between two states is invariant under conjugation by isometries.
\end{proof}

We can then use the fact that \[
	\Pr_{b_0, b_1, \measChan', \measChan}[ \wtbOne = b_1 \mid \wtbOne \gets \measChan'(\measChan(\racEnc{b_0 b_1}))] \leq \half + T(\sigma_0, \sigma_1)
\]
and that any auxiliary state independent of $b_1$ will not change the trace distance to show the following corollary for a channel, $\measChan'$, guessing $b_1$.
\begin{corollary}
	\label{cor:guess}
	Let $\Pr[b_0 = \wtb \mid \wtb \gets \calM(\racEnc{b_0, b_1})] \geq \pBad$.
	Then, for a channel $\measChan'$, $\aux_0$ independent of $b_0$ and $b_1$, and $\aux_1$ independent of $b_1$,
	\begin{align*}
		\Pr_{b_0, b_1, \measChan, \measChan'}\bigg[
			c_1 = b_1 
			\mid 
			\wtbOne \gets \measChan'(\measChan(\racEnc{b_0 b_1} \otimes \aux_0) \otimes b_0 \otimes \aux_1) ,
		\bigg] \leq \half (1 + T(\sigma_0, \sigma_1)) \leq \guessProb.
	\end{align*}
\end{corollary}

\newcommand{\fuzzyLock}{\texttt{FuzzyLock}}
\newcommand{\secL}{\lambda'}
\newcommand{\PGood}{\calP_{good}}
\newcommand{\PGeq}{\calP_{\geq \pBad}}
\newcommand{\PLt}{\calP_{< \pBad}}
\newcommand{\fracGood}{\frac{3n}{5}}
\newcommand{\fracBad}{\frac{2n}{5}}
\newcommand{\fracGoodMult}{\frac{5}{4}}
\newcommand{\totalPGoodSize}{\frac{3n}{4}}
\newcommand{\badContrUpper}{\frac{n}{2}}

\newcommand{\SBad}{\calP_{bad}}
\newcommand{\sizeBad}{\frac{n}{5}}

\section{One-Time Memory Protocol}
We now present the main result of this paper: a one-time memory protocol outlined in protocol~\ref{fig:otm} using QRACs and classical access VBB obfuscation.
We also outline a helper function in program~\ref{prog:fuzzy} which ``unlocks'' a message if the input is sufficiently close to some random string $r$.

\begin{program}[H]
	\begin{mdframed}
		$\fuzzyLock_{r, m}(r')$ for $r' \in \set{0, 1}^n$:
		\begin{itemize}
			\item If $|r' \oplus r| \leq (1 - \pAcc) n$, output $m$ 
			\item Otherwise, output $\bot$
		\end{itemize}

	\end{mdframed}
	\caption{The fuzzy lock function}
	\label{prog:fuzzy}
\end{program}

\begin{protocol}
	\begin{mdframed}
		$\OTMPrep(n, m_0, m_1)$:
		\begin{itemize}
			\item Sample $r_0, r_1 \randomGets \set{0, 1}^n$
			\item Let $\rho_i = \racEnc{r_0[i] \; r_1[i]}$ and $\rho = \bigotimes_{i \in [n]} \rho_i$
			\item Let $\orac_0$ be the oracle which accepts classical queries and outputs $\fuzzyLock_{r_0, m_0}$ and $\orac_1$ be the oracle which accepts classical queries and outputs $\fuzzyLock_{r_1, m_1}$
			\item Return $\orac_0, \orac_1, \rho$
		\end{itemize}
		$\OTMRead(\orac_0, \orac_1, \rho, \alpha)$ for $\alpha \in \set{0, 1}$
		\begin{itemize}
			\item Measure $\rho$ in $\mu_\alpha$ as defined in \cref{subsec:qrac} to get $r' \in \set{0, 1}^n$
			\item Return $\orac_\alpha(r')$
		\end{itemize}
	\end{mdframed}
	\caption{One time memories from QRACs and VBB obfuscation with only classical inputs}
	\label{fig:otm}
\end{protocol}

\subsubsection*{Correctness}
\begin{theorem}[Correctness]
	\label{thm:otmCorrect}
	For all $m_0, m_1 \in \Msgs$
	\[
		\Pr[\OTMRead(\OTMPrep(n, m_0, m_1), 0) = m_0] = 1 - \negl(n)
	\]
	and
	\[
		\Pr[\OTMRead(\OTMPrep(n, m_0, m_1), 1) = m_1] = 1 - \negl(n).
	\]
\end{theorem}
\begin{proof}
	The proof follows from a direct inspection of the protocol.
	Measuring $\rho_i$ in $\OTMRead$ according to $\mu_\alpha$ will return $r_\alpha[i]$ with probability $\cos^2(\pi / 8) > \pAcc$.
	Then, we can use a Chernoff bound
	\footnote{Note that independence follows from the honest evaluation of the protocol: each $\rho_i$ is an independent QRAC state and each measurement is done independently on each qubit.}
	to show that the probability of measuring $r'$ with $|r' \oplus r| > (1 - \pAcc) \cdot n$ is exponentially small in $n$.
\end{proof}

\subsubsection*{Soundness}
Let $p_i$ be the probability \emph{over $r_0, r_1$}, and \emph{the algorithm} that the $i$-th bit of the output of channel $\measChan$ is equal to the $i$-th bit of $r_\alpha$ given the QRAC states for all other bits for $\alpha \in \set{0, 1}$.
Formally,
\begin{align*}
	p_i = \Pr_{r_0, r_1, \calM}\left[x[i] = r_\alpha[i]  \mid x \gets \measChan\left(\bigotimes_{i \in [n]} \racEnc{r_0[i] r_1[i]}\right) \right].
\end{align*}
Now, before proving soundness, we first must show that $p_i$ must be sufficiently high for some subset of $[n]$ in order for $\calM$ to output a string close to $r_\alpha$.
Moreover, we must lower-bound the size of ``successful'' guesses for $r_\alpha[i]$ where the probability of success is at least $\pBad$.
We then use this lemma to show that the distinguishing probability for bits of $r_{1 - \alpha}[i]$ equaling $0$ or $1$ is upper-bounded away from $\pAcc$.


\begin{lemma}[First Accepting Input]
	\label{lemma:accInp}
	Fix $\alpha \in \set{0, 1}$.
	Let $x$ be the output of $\measChan$ for a fixed $r_0, r_1$.
	Let
	\[
		\PLt = \set{ i \mid p_i < \pBad }.
	\]
	Then, if $|\PLt| \geq \fracBad$,
	\begin{equation}
		\label{eq:accInp}
		\Pr_{r_0, r_1, \measChan}[\orac_\alpha(x) \neq \bot \mid x \gets \measChan(\racEnc{r_0 r_1})] \leq \negl(n).
	\end{equation}
\end{lemma}
\begin{proof}
	First note that in \cref{eq:accInp}, the input to $\measChan$ is the tensor product of independent QRAC states.
	Thus, each state is independent of the other states and so the probability of measuring $r_\alpha[i]$ is upper-ounded by the optimal QRAC probability of success, $\pMax$, for all $i \in [n]$.
	Assuming that $|\PLt| \geq \fracBad$, we have that
	\[
		\sum_{i \in [n]} p_i \leq \pMax \cdot \fracGood + \pBad \cdot \fracBad \leq 0.845n.
	\]
	Let $X_i$ be the indicator random variable that $x[i] = r_\alpha[i]$.
	Then, by the definition of $p_i$ and the super-martingale tail bound in \cref{lemma:superTail}, we have that
	\[
		\Pr\left[\sum_i X_i \geq \pAcc \right] \leq 
		\Pr\left[\sum_i X_i - \sum_i p_i \geq \pAcc - \sum_i p_i \right] \leq \exp\left(
			-\frac{(\pAcc n - 0.845n)^2}{2n}\right) \leq \negl(n).
	\]
	And because $\Pr\left[\sum_i X_i \geq \pAcc n\right] \leq \negl(n)$, we can see that $\Pr[\orac_\alpha(x) \neq \bot] \leq \negl(n)$ by definition of program~\ref{prog:fuzzy}.
\end{proof}

We are now ready to prove the simulation security of the one-time memory protocol.

\begin{theorem}[Soundness]
	\label{thm:otmSound}
	For any quantum PPT adversary, $\calA$, there exists a simulator $\Sim$ for every $m_0, m_1 \in \Msgs$ such that $\Sim$ makes at most one query to $g^{m_0, m_1} : \set{0, 1} \rightarrow \set{m_0, m_1}$ (where $g(\alpha) = m_\alpha$) and
	\[
		\calA\left(\OTMPrep(1^{n}, m_0, m_1)\right) \overset{\negl(n)}{\approx_c} \SimOTM^{g^{m_0, m_1}}\left(1^{n} \right).
	\]
\end{theorem}

\begin{proof}
	We will consider a series of hybrids.
	\begin{itemize}
		\item $\Hyb_0$: The real protocol:
			\begin{enumerate}
				\item Run $\OTMPrep(n, m_0, m_1)$ to get $\orac_0$, $\orac_1$, and $\rho$. Give the adversary $\orac_0, \orac_1$, and $\rho$.
				\item Run $\calA(\orac_0, \orac_1, \rho)$ to return output $\out$
			\end{enumerate}
		\item $\Hyb_1$: The same as $\Hyb_0$ except we consider an adversary, $\calA$, which makes $q$ total calls to both $\orac_0$ and $\orac_1$ for $q \in \poly(n)$. We break up $\calA$ into
			\[
				\calA = \calS_q \circ \oracP_{q - 1} \circ \calS_{q - 1} \circ \ldots \circ \oracP_1 \circ \calS_{1}
			\]
			where $\calS_j$ is a (quantum) PPT algorithm and the oracle $\oracP_j = (\orac_0, \orac_1, \Iden)$ where $\Iden$ is the identity channel to allow $\calS_j$ to pass state to $\calS_{j + 1}$.
			Moreover, we restrict $\calS_j$ to send only one message between both $\orac_0$ and $\orac_1$.
		\item $\Hyb_2$: Let $\ell \in [q]$ such that $\calS_\ell$ is the first algorithm which sends $x$ to $\oracP_\ell$ such that $\orac_\alpha(x) \neq \bot$ for either $\alpha = 0$ or $\alpha = 1$ with non-neglible probability.
			Replace $\oracP_w$ for $w \in [\ell - 1]$ with $\oracP_w' = (\bot, \bot, \Iden)$ where $\bot$ is the null oracle.
		\item $\Hyb_3$:  For a fixed $\alpha$, let $\ell' \in [q]$ with $\ell' > \ell$ such that $\calS_{\ell'}$ is the first algorithm to send $x$ to $\oracP_{\ell'}$ such that $\orac_{1 - \alpha}(x) \neq \bot$ with non-neglible probability.
			If no algorithm makes an accepting call to $\orac_{1 - \alpha}$ with non-neglible probability, then let $\ell' = q$.
			Then, for all $w \in \set{\ell, \ell + 1, \dots, \ell' - 1}$, replace $\oracP_w$ with $\oracP_w'' = (\orac_0, \orac_\bot, \Iden)$ if $\alpha = 0$ and $\oracP_w'' = (\orac_\bot, \orac_1, \Iden)$ where $\orac_\bot$ is the null oracle.
	\end{itemize}
	The first hybrid holds because $\Hyb_1$ is equivalent to $\Hyb_0$ as $\calA$ makes calls to oracles $\orac_0, {\orac_1}$ in $\Hyb_1$ and maintains internal state between each call.
	We can thus model $\calA$ as a set of (quantum) PPT algorithms between each call which can pass arbitrary state to the next algorithm. 
	For $\Hyb_2$, we assume that $\calS_{w}$ for $w \in [\ell - 1]$ does not make an accepting call to $\orac_0$ or $\orac_1$ with non-neglible probability.
	So, we replace the oracle calls with the null oracle.
	Similarly, for $\Hyb_3$, we assume that $\calS_{w}$ for $w \in [\ell, \ell' - 1]$ does not make an accepting call to $\orac_{1 - \alpha}$ with non-neglible probability and can thus be replaced by the null oracle.

	We can now outline our basic simulator, $\SimOTM^{g^{m_0, m_1}}(1^{n})$:
	\begin{itemize}
		\setlength{\itemsep}{0.1em}
		\item Run $g(\alpha)$ to get $m_\alpha$.
		\item Run $\OTMPrep(n, m_0', m_1')$ for $m'_\alpha = m_\alpha$ and $m'_{1 - \alpha} \randomGets \Msgs$ to get $\overline{\orac}_\alpha, \overline{\orac}_{1 - \alpha}, \overline{\rho}$ where $\overline{\orac}_{1 - \alpha}$ is the null oracle and $\overline{\orac}_\alpha$ is the same as in the protocol.
		\item Call $\advers(\overline{\orac}_0, \overline{\orac}_1, \overline{\rho})$ to get output $\out$.
	\end{itemize}

	To prove the indistinguishability of the adversary's view in $\Hyb_3$ and the simulator,
	we have two cases: if $\ell = q$, then $\calA$ never makes an accepting call to either $\orac_0$ or $\orac_1$ and so, $\orac_{1 - \alpha}$ is indistinguishable from the null oracle.
	So then, the simulator's oracle is indistinguishable from the real oracle.

	For the second case, assume that $\ell < q$. 
	Without loss of generality, assume that $\alpha = 0$ and that $\Sim$ makes a call to $g(0)$ (and thus learns $m_0$).
	We now want to show that with overwhelming probability, no accepting queries to $C_1$ are made by $\calS_\ell, \dots \calS_{q - 1}$ and so $\ell' = q$.
	Then, if $\ell' = q$, we have
	\[
		\calS_q \circ \calO''_{\ell} \circ \ldots \calO''_{\ell} \circ \calA_{\ell} \circ \ldots \circ \calO'_{\ell - 1} \circ \ldots \circ \calO'_1 \circ \calS_1\left(\OTMPrep(1^{n}, s_0, s_1)\right)
	\]
	and so $m_1$ is removed from the view of $\calS$ and $\orac_1$ is indistinguishable from the null oracle.

	Now, for the second case, we need to show that 
	\[
		\Pr[|r_1 - x| \geq \pAcc n \mid x \gets (S_{q - 1}, \dots S_\ell) ] \leq \negl(n).
	\]
	We first make use of \cref{lemma:accInp}.
	Note that for the first $\ell - 1$ calls, $\orac_0$ and $\orac_1$ are replaced with the null oracle.
	So, in order for $\calS_\ell$ to make an accepting call to $\orac_0$, it must be the case that $|\PGeq| \geq \fracGood$ where
	\[
		\PGeq = [n] \setminus \PLt
	\]
	as defined in \cref{lemma:accInp}.
	Also, note that for $w' \in \set{\ell, \dots, \ell'}$, $\calS_{w'}$ gains information about $m_0$ and $r_0$ from $\orac_0$ while gaining no auxilary information about $r_1$ or $m_1$.
	Thus, any auxilary information learned from the circuit in $\orac_0$ is independent of $r_1$ and $m_1$.
	Let $\measChan$ be the quantum channel on $\rho$ from the first $\ell$ algorithms and $\measChan'$ be any subsequent (quantum) strategy for guessing a string close to $r_1$ prior to the $\ell'$-th call to the oracle.

	Thus, by \cref{claim:nonconv} and \cref{cor:guess}, for all $i \in \PGeq$,
	\[
		\Pr_{\calM, \calM', r_0, r_1}\left[
			r_1[i] = c[i] \mid c[i] \gets \measChan'\left(\measChan\left(\bigotimes_i \racEnc{r_0[i] r_1[i]}\right)
				\otimes r_0 \otimes m_0 
			\right) 
		\right] \leq \guessProb.
	\]
	We can again invoke the super-martingale tail bound in \cref{lemma:superTail} to show that
	\[
		\Pr\left[
			\bigwedge_{i \in \calG} r_1[i] = c[i] \mid c[i] \gets \measChan'\left(\measChan\left(\bigotimes_i \racEnc{r_0^i r_1^i}\right) \otimes r_0 \otimes m_0 \right)
		\right] \leq \negl(n)
	\]
	for any $\calG \subseteq \PGeq$ and $|\calG| \geq \frac{7}{10} \cdot |\PGeq|$.

	So, with overwhelming probability, $\measChan'$ cannot guess any $70\%$ fraction (or more) of the bits in $\PGeq$ with non-neglible probability.

	Next, in the worst case, $\calA$ guesses $r_1[i]$ with probability $1$ for all $i \notin \PGeq$ (i.e. $i \in \PLt$)~\footnote{We suspect that the concrete bounds can be improved by considering an upper-bound of $\pMax$ instead of $1$ for the worst-case scenario though it is not clear how to do this given the auxiliary information leaked by $\orac_0$.}.
	Then note that $|\PLt| \leq \frac{2n}{5}$ and so $\measChan'$ cannot guess more that $\frac{2n}{5} + \frac{7}{10} \cdot \frac{3n}{5} \leq 0.82 n$ bits of $r_1$ with non-neglible probability.
	Because $0.82 n \leq \pAcc n$, we get
	\[
		\Pr\left[
			|r_1 - c| \geq \pAcc \cdot n \; \bigg| \; c \gets \measChan'\left(\measChan\left(\bigotimes_i \racEnc{r_0^i r_1^i}\right) \otimes r_0 \otimes m_0 \right)
		\right] \leq \negl(n).
	\]
	So finally, we have that $\advers$ does not make an accepting call to $\orac_1$ with non-neglible probability.
	So, $\orac_1$ is indistinguishable from the null oracle.
	Thus, $\SimOTM$ can simulate the view of $\calA$ as previously noted.
\end{proof}

\section*{Acknowledgments}
The author is grateful to the helpful discussions and feedback from Fabrizio Romano Genovese, Joseph Carolan, Stefano Gogioso, Matthew Coudron, and Ian Miers
LS acknowledges funding and support from the QSig Commision and from the NSF Graduate Research Fellowship Program.

\bibliographystyle{alpha}
\bibliography{bib/ref}

\newcommand{\etalchar}[1]{$^{#1}$}
\begin{thebibliography}{ANTSV99}

\bibitem[ANTSV99]{ambainis1999dense}
Andris Ambainis, Ashwin Nayak, Ammon Ta-Shma, and Umesh Vazirani.
\newblock Dense quantum coding and a lower bound for 1-way quantum automata.
\newblock In {\em Proceedings of the thirty-first annual ACM symposium on
  Theory of computing}, pages 376--383, 1999.

\bibitem[Azu67]{azuma1967weighted}
Kazuoki Azuma.
\newblock Weighted sums of certain dependent random variables.
\newblock {\em Tohoku Mathematical Journal, Second Series}, 19(3):357--367,
  1967.

\bibitem[BGS12]{Broadbent2012QuantumOP}
Anne Broadbent, Gus Gutoski, and Douglas Stebila.
\newblock Quantum one-time programs.
\newblock In {\em IACR Cryptology ePrint Archive}, 2012.

\bibitem[BGZ21]{Broadbent_2021}
Anne Broadbent, Sevag Gharibian, and Hong-Sheng Zhou.
\newblock Towards quantum one-time memories from stateless hardware.
\newblock {\em Quantum}, 5:429, April 2021.

\bibitem[BS23]{barhoush2023powerful}
Mohammed Barhoush and Louis Salvail.
\newblock Powerful primitives in the bounded quantum storage model.
\newblock {\em arXiv preprint arXiv:2302.05724}, 2023.

\bibitem[BSS21]{behera2021noise}
Amit Behera, Or~Sattath, and Uriel Shinar.
\newblock Noise-tolerant quantum tokens for mac.
\newblock {\em arXiv preprint arXiv:2105.05016}, 2021.

\bibitem[CGLZ19]{chung2019cryptography}
Kai-Min Chung, Marios Georgiou, Ching-Yi Lai, and Vassilis Zikas.
\newblock Cryptography with disposable backdoors.
\newblock {\em Cryptography}, 3(3):22, 2019.

\bibitem[GH97]{godbole1997beyond}
Anant~P Godbole and Pawel Hitczenko.
\newblock Beyond the method of bounded differences.
\newblock {\em Microsurveys in Discrete Probability}, 41:43--58, 1997.

\bibitem[GIS{\etalchar{+}}10]{goyal2010founding}
Vipul Goyal, Yuval Ishai, Amit Sahai, Ramarathnam Venkatesan, and Akshay Wadia.
\newblock Founding cryptography on tamper-proof hardware tokens.
\newblock In {\em Theory of Cryptography: 7th Theory of Cryptography
  Conference, TCC 2010, Zurich, Switzerland, February 9-11, 2010. Proceedings
  7}, pages 308--326. Springer, 2010.

\bibitem[GKR08]{goldwasser2008one}
Shafi Goldwasser, Yael~Tauman Kalai, and Guy~N Rothblum.
\newblock One-time programs.
\newblock In {\em Advances in Cryptology--CRYPTO 2008: 28th Annual
  International Cryptology Conference, Santa Barbara, CA, USA, August 17-21,
  2008. Proceedings 28}, pages 39--56. Springer, 2008.

\bibitem[Hoe94]{hoeffding1994probability}
Wassily Hoeffding.
\newblock Probability inequalities for sums of bounded random variables.
\newblock {\em The collected works of Wassily Hoeffding}, pages 409--426, 1994.

\bibitem[Kat07]{Katz2007UniversallyCM}
Jonathan Katz.
\newblock Universally composable multi-party computation using tamper-proof
  hardware.
\newblock In {\em International Conference on the Theory and Application of
  Cryptographic Techniques}, 2007.

\bibitem[Liu14]{liu2014single}
Yi-Kai Liu.
\newblock Single-shot security for one-time memories in the isolated qubits
  model.
\newblock In {\em Advances in Cryptology--CRYPTO 2014: 34th Annual Cryptology
  Conference, Santa Barbara, CA, USA, August 17-21, 2014, Proceedings, Part II
  34}, pages 19--36. Springer, 2014.

\bibitem[Liu15]{liu2015privacy}
Yi-Kai Liu.
\newblock Privacy amplification in the isolated qubits model.
\newblock In {\em Advances in Cryptology-EUROCRYPT 2015: 34th Annual
  International Conference on the Theory and Applications of Cryptographic
  Techniques, Sofia, Bulgaria, April 26-30, 2015, Proceedings, Part II 34},
  pages 785--814. Springer, 2015.

\bibitem[Liu23]{liu2023depth}
Qipeng Liu.
\newblock Depth-bounded quantum cryptography with applications to one-time
  memory and more.
\newblock In {\em 14th Innovations in Theoretical Computer Science Conference
  (ITCS 2023)}, 2023.

\bibitem[LLPY23]{li2023classical}
Xingjian Li, Qipeng Liu, Angelos Pelecanos, and Takashi Yamakawa.
\newblock Classical vs quantum advice and proofs under classically-accessible
  oracle.
\newblock {\em arXiv preprint arXiv:2303.04298}, 2023.

\bibitem[Nai43]{naimark1943representation}
Mark~A Naimark.
\newblock On a representation of additive operator set functions.
\newblock In {\em Dokl. Akad. Nauk SSSR}, volume~41, pages 373--375, 1943.

\bibitem[NC01]{nielsen2001quantum}
Michael~A Nielsen and Isaac~L Chuang.
\newblock {\em Quantum computation and quantum information}, volume~2.
\newblock Cambridge university press Cambridge, 2001.

\bibitem[PRM09]{Perrig2009ReducingTT}
Adrian Perrig, Michael~K. Reiter, and Jonathan~M. McCune.
\newblock Reducing the trusted computing base for applications on commodity
  systems.
\newblock 2009.

\bibitem[Smi81]{smid1981integrating}
M~Smid.
\newblock Integrating the data encryption standard into computer networks.
\newblock {\em IEEE Transactions on Communications}, 29(6):762--772, 1981.

\bibitem[SYG{\etalchar{+}}19]{skarlatos2019microscope}
Dimitrios Skarlatos, Mengjia Yan, Bhargava Gopireddy, Read Sprabery, Josep
  Torrellas, and Christopher~W Fletcher.
\newblock Microscope: Enabling microarchitectural replay attacks.
\newblock In {\em Proceedings of the 46th International Symposium on Computer
  Architecture}, pages 318--331, 2019.

\bibitem[VGO{\etalchar{+}}20]{2020SciPy-NMeth}
Pauli Virtanen, Ralf Gommers, Travis~E. Oliphant, Matt Haberland, Tyler Reddy,
  David Cournapeau, Evgeni Burovski, Pearu Peterson, Warren Weckesser, Jonathan
  Bright, St{\'e}fan~J. {van der Walt}, Matthew Brett, Joshua Wilson, K.~Jarrod
  Millman, Nikolay Mayorov, Andrew R.~J. Nelson, Eric Jones, Robert Kern, Eric
  Larson, C~J Carey, {\.I}lhan Polat, Yu~Feng, Eric~W. Moore, Jake
  {VanderPlas}, Denis Laxalde, Josef Perktold, Robert Cimrman, Ian Henriksen,
  E.~A. Quintero, Charles~R. Harris, Anne~M. Archibald, Ant{\^o}nio~H. Ribeiro,
  Fabian Pedregosa, Paul {van Mulbregt}, and {SciPy 1.0 Contributors}.
\newblock {{SciPy} 1.0: Fundamental Algorithms for Scientific Computing in
  Python}.
\newblock {\em Nature Methods}, 17:261--272, 2020.

\bibitem[Yao82]{Yao1982ProtocolsFS}
Andrew Chi-Chih Yao.
\newblock Protocols for secure computations.
\newblock {\em 23rd Annual Symposium on Foundations of Computer Science (sfcs
  1982)}, pages 160--164, 1982.

\bibitem[ZCD{\etalchar{+}}19]{Zhao2019OneTimePM}
Lianying Zhao, Joseph~I. Choi, Didem Demirag, Kevin R.~B. Butler, Mohammad
  Mannan, Erman Ayday, and Jeremy Clark.
\newblock One-time programs made practical.
\newblock {\em ArXiv}, abs/1907.00935, 2019.

\end{thebibliography}

\appendix

\section{Tail Bounds}
\label{sec:tailbounds}

\begin{lemma}[Supermartingale Probability Tail Bounds]
	\label{lemma:superTail}
	Let $X_i$ for $i \in [n]$ be an indicator random variable that $f_i(x_i) = 1$ for some $f_i \randomGets F_i$ and set of functions $F_i$.
	Then, if
	\[
		\Pr_{x_i, f_i}[f_i(x_i) = 1 \mid x_1, \dots x_n] \leq p_i
	\]
	we have
	\[
		\Pr\left[
		\sum X_i - \sum p_i \geq t \right] \leq \exp\left(
			-\frac{t^2}{2n}
		\right)
	\]
\end{lemma}
\begin{proof}
	Let $S_i = \sum_{j = 1}^i X_j$ and $S_0 = 0$.
	Then, let $Y_i = S_i - \sum_{j = 1}^i p_j$ where $Y_0 = 0$.
	Because $\E[X_i \mid X_1, \dots, X_{i - 1}, X_{i + 1}, X_n] \leq p_i$, we have
	\(
	\E[X_i \mid X_1, \dots X_{i - 1}] \leq p_i
	\)
	and so
	\begin{align*}
		\E[Y_{i + 1} \mid Y_1 \dots Y_i] &\leq S_i + \E[X_{i + 1}] - \sum_{j = 1}^{i} p_j - p_{i + 1} \\
						 &\leq S_i - \sum_{j = 1}^i p_j = Y_i.
	\end{align*}
	We thus have that $Y_0, \dots, Y_n$ is a supermartingale \cite{godbole1997beyond}.
	We can then note that $|Y_i - Y_j| \leq 1$ and so, using Azuma-Hoeffding's inequality for supermartingales \cite{azuma1967weighted, hoeffding1994probability},
	\[
		\Pr\left[
		Y_n  - Y_0 \geq t \right] \leq \exp\left(
		-\frac{t^2}{2n}\right).
	\]
	Because $Y_n = S_n - \sum_{i = 1}^n p_i$, we have that
	\[
		\Pr\left[
		S_n - \sum_{i = 1}^n p_i \geq t \right] \leq \exp\left(
		-\frac{t^2}{2n}
	\right)
\]
as desired.
\end{proof}
%

\end{document}